%% file: SIN.tex
\documentclass{amsart}
\usepackage[authoryear,sort&compress]{natbib}
\usepackage{graphicx}

\usepackage{epsfig}
\usepackage{latexsym}
\usepackage{amsfonts}
\usepackage{amsthm}
\usepackage{epsfig}
\usepackage{latexsym}
\usepackage{mathrsfs}
\usepackage{epsfig}
\usepackage{latexsym}
\usepackage{epic}
\usepackage[leqno]{amsmath}
\usepackage{amssymb}
   \newenvironment{myindentpar}[1]%
   {\begin{list}{}%
            {\setlength{\leftmargin}{#1}}%
            \item[]%
   }
   {\end{list}}
\newtheorem{lemma}{Lemma}[section]
\newtheorem{theorem}{Theorem}

\newtheorem{example}[lemma]{Example}

\newcommand{\Pb}{\mathbb P}

\newcommand{\M}{\mathcal M}

\newcommand{\Tint}{\mathring{\Theta}}
\newcommand{\Tova}{\overline{p(\Theta(a))}}
\newcommand{\Tovb}{\overline{p(\Theta(b))}}

\newcommand{\PP}{{\mathbb P}}
\newcommand{\EE}{{\mathbb E}}
\newcommand{\RR}{{\mathbb R}}

\begin{document}
\raggedright

\begin{center}
\textbf{Can we avoid `SIN' in the  house of `No Common Mechanism'?}\\
\textsc{Mike Steel}
\end{center}

Author's Affiliations\\
MIKE STEEL\\
Allan Wilson Centre for Molecular Ecology and Evolution, \\Biomathematics Research Centre, University of Canterbury, Private Bag 4800, Christchurch, New Zealand; \\
E-mail: m.steel@math.canterbury.ac.nz\\


\begin{center}
{\bf Abstract}
\end{center}

In `no common mechanism' (NCM) models  of character evolution, each character can evolve on a phylogenetic tree under a partially or totally separate process (e.g. with its own branch lengths). 
 In such cases, the usual conditions that suffice to establish the statistical consistency of
tree reconstruction by methods such as maximum likelihood (ML) break down, suggesting that such methods may be prone to statistical inconsistency (SIN). In this paper we ask whether we can avoid SIN
for tree topology reconstruction when adopting such models, either by using ML or any other method that could be devised. We prove that it is possible to avoid SIN for certain NCM models, but not for others, and the results depend delicately on the tree reconstruction method employed.  We also describe the biological relevance of some recent mathematical results for the more usual `common mechanism' setting.  Our results are not intended to justify NCM, rather to set in place a framework within which such questions can be formally addressed.

\newpage

\section{Introduction}

Statistical Inconsistency (hereafter, SIN) in phylogenetics is the tendency of certain tree reconstruction methods to converge on an incorrect tree topology when applied to
increasing quantities of data that evolve under a given model.  The phenomenon has been well known for simple methods like maximum parsimony since the landmark paper of  \cite{fels} three decades ago. SIN has contributed
to the widespread acceptance of more sophisticated tree reconstruction methods such as maximum likelihood, corrected distance methods and Bayesian phylogenetics \citep{fel}, \citep{handbook}. These methods
are based explicitly on stochastic models of sequence evolution, and for which it is usually possible to establish statistical consistency when the model assumed by the investigator is also the one that  generated the data (see, for example, \cite{cha}, \cite{all}, \cite{sober2}).

A centrepiece of nearly all these models is the assumption that  character data (for instance, genetic sequence sites) evolve independently and identically. This `i.i.d.' assumption is standard in statistics, and implies that each character is described by essentially the same process and that the characters represent a finite random sample of this process.   This i.i.d. assumption applies even for mainstream models that allow
a distribution of rates across sites, such as the frequently used `Gamma+I' embellishment of the general time reversible (GTR) model. In these models, it is usually assumed that the rate at a site is chosen i.i.d. from a given distribution.  Such a `rates-across sites' model is subtly different from a model that assumes that each site has its own particular intrinsic rate (i.e. not chosen i.i.d. from some distribution) -- let us call it a `variable site rate' model.  Within such a model, the sequence sites  may still be independently generated, but they are not identically distributed (some sites simply {\em are} evolving faster than others).

If we just consider the frequencies of site patterns,  then the two models (rates across sites, and variable site rate) can produce (almost) identical data;  however, significant differences between the models can become apparent when we come to do tree reconstruction from given sequences. For example, in a maximum likelihood approach to tree reconstruction, in which we explicitly assume the variable site rate model we may wish to estimate a corresponding rate for each site that maximizes the probability of observing the given site pattern -- along with a shared underlying set of branch lengths common to all the sites (such an approach was described by Gary Olsen in \cite{swo}).  Each rate estimate - one for each site - might later be discarded as a `nuisance parameter' in the search for the underlying tree topology alone.   

This approach is quite different to doing the `usual' form of maximum likelihood estimation of a tree topology under a rates-across-sites model. We can ask if such an approach is statistically sound - in particular, can it lead to SIN?  What if we allow the branch lengths also to vary from character to character (the more usual form of  `no common mechanism')?  Is maximum likelihood under this model liable to SIN; if so, can any method reconstruct a tree under this model without SIN? These are the sort of questions we will address.  We will also describe the biological relevance of some recent mathematical results concerning tree reconstruction in the more usual `common mechanism' setting.

First we outline some of the motivations and concerns surrounding no common mechanism models in phylogenetics. We then discuss statistical consistency in a general setting -- first for common mechanism models, where much is known, then for no-common mechanism models, where there has been little analysis to date in phylogenetics.  In Theorem 1 we present some first results in this area and show how the details of the model (and the method) are crucial  to whether we are in danger of SIN when working with a no-common mechanism model.   We also describe different forms of SIN, and attempts to measure and manage it. The paper ends with a brief  discussion.

\section{Some reasons for and against `No common mechanism'?}

The idea that the evolution of characters in biology might be described by different sets of branch lengths underlies recent attempts to deal with phenomena such as heterotachy \citep{gau, phil}.  However, the idea dates back to the early days of molecular phylogenetics.  It is implicit in Walter Fitch's discussion of a covarion  model  \citep{fitch}, and  was discussed more explicitly  by  \cite{cav} in reference to his  simple two--state Poisson model.  In response to the question of whether the probabilities of change should be the same for all characters, Cavender remarked:
 \begin{myindentpar}{1cm}
 ``This assumption can and should be removed. It is unacceptable biologically because it says, for example, that an insect species is just as likely to lose (or acquire) wings as a spot of color.''
   \end{myindentpar}
   This comment seems reasonable for morphological characters, though even in that setting one might still expect some correlation in the relative probabilities of character change on a given branch across characters, as it may be more likely to observe changes on branches that correspond to long time intervals between speciation.  It is less clear that Cavender's comment should apply to aligned DNA sequence sites, each of which we might view as a random samples from a common process. Nonetheless different DNA sequence sites may be subject to differing selection pressures and the probability that a site mutation becomes fixed in a population may depend on structural or functional constraints; for example, whether the protein a gene codes for still folds correctly if the substitution changes an amino acid. These constraints  may vary with time and across the sequence, so enforcing an entirely `common mechanism' model may be too severe.  Similar comments apply to other types of genomic data that carry evolutionary signal.  In  linguistics, a model that allow each character to have its own branch lengths has also been developed for studying language evolution \citep{war}.

An additional reason why the No Common Mechanism (hereafter NCM) approach has received further attention is its relevance to those in the systematics community who advocate the use of maximum parsimony (herafter MP) for phylogeny reconstruction (e.g. \cite{far}).   This has been justified by an equivalence theorem  that demonstrates that MP is the maximum likelihood (hereafter ML) estimator of a tree under a NCM  model based on a symmetric Poisson process such as the Jukes-Cantor model \citep{tuf}.  A slightly more streamlined proof of this result  has recently been given by \cite{fischer}) and extensions of this equivalence theorem were described in \cite{penste1},  \cite{penste2}, and, most recently, \cite{fischer}. This last paper also showed  that the original equivalence theorem breaks down if one modifies the Poisson model slightly; either (i) by imposing a molecular clock, or (ii) by setting an absolute upper bound on the branch lengths.

The significance and implications of the equivalence between MP and ML estimation under NCM have aroused considerable interest (see, for example, \cite{so}, \cite{far}, \cite{hue2, hue}).  One view is that NCM model is sufficiently general as to capture `truth' and so should be the model of choice, thereby providing a justification for maximum parsimony (\cite{far}).  An alternative position is that NCM is far too parameter rich and  it ignores likely correlations between branch lengths due to
shared time frames of speciation intervals.  The NCM model required for the formal equivalence between MP and ML under the NCM is also based on a symmetric model of substitution change (such as a Jukes-Cantor model). Note that this model predicts (approximately) equal base frequencies, however a formal equivalence between MP and ML under the NCM model still holds if one regards the ancestral base in the tree at each site as a further parameter to be estimated (this would allow each site to have a `preferred' base, to reflect observed base composition in sequences). 

 In the sections that follow our aim is not to defend NCM models, but rather to determine which methods, if any, would allow phylogenetic tree topology to be estimated in a statistically consistent manner were one  to 
 adopt  various NCM models.

\section{ML estimation in general and in phylogenetics}

In this section we consider a general setting that includes phylogenetic tree reconstruction, and other problems where a discrete parameter (e.g. a tree, network, cluster) is being estimated from discrete data (e.g. DNA sequences, genes) in the presence of unknown additional parameters.

Suppose we have a sequence of observations $u_1, u_2, \ldots$ taking values in a finite set $U$ (the elements of this set can be arbitrary, but we will call them `site patterns' as we will usually be considering aligned DNA sequence sites). Suppose that these observations are generated independently by a model $M$ that has a fixed but unknown discrete parameter  $a$ that takes values in some finite set $A$, alongside other continuous parameters which may vary from observation to observation.  In the phylogenetic setting, $A$ will generally refer to the set of fully-resolved tree topologies on a given set of species, and the continuous parameters may refer to branch lengths or other aspects of the substitution model (site rate, transition/transversion ratio, shape parameter for a $\Gamma$ distribution of rates across sites etc).

In all such cases, $u_i$ is generated by a pair $(a, \theta_i)$ where $\theta_i$ lies in some  set $\Theta(a)$ which we will assume throughout is an open subset of Euclidean space. In the case of branch lengths on  a tree this means they should be strictly positive but finite real numbers.

\subsection{CM and NCM versions of a model}
In the Common Mechanism version of a model $M$, which we will denote by CM--$M$, it is
assumed that all the $\theta_i$ values are equal; that is, they take a common value, $\theta \in \Theta(a)$.  By contrast, in the No Common Mechanism version of $M$, which we will denote by NCM--$M$, the $\theta_i$ can take different values. Notice, however, that if these $\theta_i$ values are assigned randomly and independently from some common distribution (as is the case with most `rates across sites' models in phylogenetics) then this is just a CM version of a slightly more complex model $M^*$ that is derived from $M$.

\subsection{Maximum Likelihood under CM and NCM}

The ML estimation of a discrete parameter from $A$ under an NCM version of $M$ applied to data $(u_1, \ldots, u_k)$ selects the element $b \in A$ that maximizes
\begin{equation}
\label{like}
L(b|{\rm data}) := \sup_{(\theta_1, \ldots, \theta_k) \in \Theta(b)^k} \Pb[{\rm data}|b, (\theta_1, \ldots, \theta_k)] = \prod_{i=1}^k \sup_{\theta \in  \Theta(b)}\Pb[u_i |b, \theta],
\end{equation}
where  `sup' in (\ref{like}) refers to supremum (the maximum value either obtained or as a limit).   The second equality in
(\ref{like})  is justified by the assumption that the observations are independently generated by the model. For ML estimation under the CM version of $M$, the only difference is that the $\theta_i$ values are required to be identical (i.e. $\theta_i = \theta$ for all $i$).

Given two models $M_1$ and  $M_2$ (usually, but not necessarily the same model),  we refer to {\em ML estimation under $M_1$ applied to $M_2$--data} as the ML estimation under model $M_1$ of the discrete parameter in $A$  from data that has been produced under model $M_2$.
We are interested in determining when this method is statistically consistent (defined shortly) for various $M_1, M_2$, and if so, what can be said about the sequence length requirements for accurate estimation.

Given two models $M_1$ and $M_2$, we write $M_1 \subseteq M_2$ if $M_1$ is a submodel of $M_2$, that is, $M_1$ is a special case of $M_2$ once constraints are placed on its parameters; in particular,  for any model $M$,
CM--$M$ $\subseteq$ NCM--$M$.   If $M_2$ is not contained in $M_1$, the ML estimator is often said to be carried out under a `mis-specified model'  - in this case, we do not generally expect consistency so we are usually more interested in the regular case where the model in which ML is performed includes the model that generates the data, that is,  either $M_1 = M_2$ or $M_2 \subseteq M_1$ (one
exception occurs in Theorem~\ref{mainthm}(iv)  which provides an instance where ML estimation under a CM  model applied to a NCM version of that model is statistically consistent).

\subsection{Basic models for character evolution ($N_r$)}

It will be convenient to describe most of our results for a particular model of character evolution. The simplest such model assumes that the rates of substitution between each pair of the $r$ character states are equal -- this is sometimes referred to as the Neyman $r$--state model or the `symmetric r-state model';  here we call it the $N_r$--model (after the Neyman $r$--state model).
  In the special case where $r=4$, this is the familiar Jukes-Cantor model, while $r=2$ is often referred to as the
`Cavendar-Farris-Neyman (CFN) model'.  In the $N_r$ model, it  is usually (but not always) assumed that the frequency of bases at the root of the tree is the uniform distribution.

We will also consider the limiting case of the $N_r$ model as $r$ becomes large (for a given number of species).  This model, denoted here by $N_\infty$, is sometimes called the `Kimura-Crow infinite alleles model' \citep{kim2} or the `Random Cluster model' \citep{mossy}, and it models the setting where each substitution always results in a new state.
We will denote the CM and NCM versions of $N_r$ model ($r$ being either finite or infinite) by writing CM--$N_r$ and NCM--$N_r$, respectively.

\subsection{SIN for data generated under Common Mechanism (CM) models}

 In the `common-mechanism' (CM) model -- either for generating the data or for carrying out ML -- we require the  $\theta_i$ values to all be equal to some common value (call it $\theta$).   Note that even if we are not at all   interested in estimating the $\theta$ values, we often still have to consider their role in any probability calculations;   in this case, they are said to be `nuisance parameters'.

 A method $\M$  for estimating the discrete parameter in $A$  from a sequence  of independently generated observations is {\em statistically consistent} for data generated under a CM model if, for each $a \in A$, and $\theta \in \Theta(a)$, the probability that $\M$  correctly estimates $a$ from $(u_1, \ldots, u_k)$, when each observation $u_i$ is generated independently by the model with parameters $(a, \theta)$,  converges to $1$ as $k$ grows.
 If this condition fails, the method $\M$ leads to statistical inconsistency (SIN).  A related, but slightly different concept of statistical consistency exists, based on the strong (rather than the weak) law of large numbers, but we do not discuss it here.

Two types of SIN are possible in inferring phylogenetic tree topology. The more familiar and stronger form is when the method $\M$ can `positively mislead' -- that is, the probability that the method estimates an incorrectly resolved tree converges to $1$ as the sequence length grows; this is the type of inconsistency that occurs  with maximum parsimony in the `Felsenstein Zone' \citep{fels}.

However, a milder form of SIN occurs if, with increasing data, the method becomes unable to decide between the true tree and alternative trees. This can occur if the method returns a non-resolved tree, of which the true tree is just one resolution, and the probability of returning such a non-resolved tree from data generated under the CM model tends to $1$ as the sequence length grows.  Precisely such a situation has recently been shown to occur with `Ancestral Maximum Likelihood'  (AML). In a maximum-likelihood framework this method optimizes not just the tree topology and its  branch lengths but also a particular set of ancestral sequences, and then returns just the tree topology.  \cite{mos} showed that this AML estimation of tree topology applied to CM--$N_2$ data  can converge on the fully-unresolved star tree, when  the branch lengths of the fully resolved generating tree are in a certain range.  Whether AML can lead to the stronger form of SIN of being positively misleading is currently an open question.

Note that this milder form of SIN is quite different from not having sufficient data to resolve a tree topology (a much more familiar problem for biologists) -- we deal with this latter issue in Section~\ref{measure}. By contrast, mild SIN requires that the tree will never become fully resolved, no matter how much data we were to obtain. 

\subsection{Topological aspects of statistical estimation}

We now describe two conditions (`Identifiability' and `Kissing') that make accurate estimation of the discrete parameter $a \in A$ simultaneously possible and problematic in the following sense: Given `enough' data we can be sure to reconstruct $a$ correctly,  but we cannot say in advance how large `enough' will be. These two conditions typically hold in the reconstruction of fully-resolved phylogenetic trees as well as other related problems. To describe the conditions we need two further definitions.

 Given the model parameters $(a, \theta)$ let $p_{(a,\theta)}$ denote the associated probability distribution on site patterns, and let
$p(\Theta(a)) := \{p_{(a,\theta)}: \theta \in \Theta(a)\},$ which is a subset of the $|U|$--dimensional simplex of probability distributions on  the set $U$ of site patterns.
Also, given a subset $A$ of Euclidean space, let $\overline{A}$ denote its (topological) closure. 
We can now state the two conditions:   For all $a, b \in A$, with $a \neq b$  consider the following:

\noindent {\bf Identifiability condition:}
\begin{equation}
\label{ideeq}
p(\Theta(a)) \cap \Tovb  = \emptyset, \mbox{ and }
\end{equation}
\noindent {\bf Kissing condition:}
 \begin{equation}
 \label{kiseq}
 \Tova \cap \Tovb  \neq \emptyset.
 \end{equation}
In the phylogenetic setting, where we will often regard $\Theta$ as branch lengths, $p(\Theta(a))$ will be all the probability distributions we can obtain on site patterns by varying the branch lengths on the binary tree $a$ over all strictly positive but  positive values. The set $\Tovb$ includes not just  all probability distributions we can obtain on site patterns by varying the branch lengths on the binary tree $b$ over all strictly positive but  finite values, but also the limiting distributions as branch lengths tend to zero or to infinity (in all possible combinations).  We provide an example (and figure) to illustrate these concepts after some brief remarks.

In general, the identifiability condition (\ref{ideeq}) alone is sufficient to ensure that maximum likelihood in the CM setting will consistently reconstruct each discrete parameter in $A$ when the observations are generated under a common mechanism \citep{cha, inv3}.   The condition holds for many models in molecular systematics, including the general Markov model, a simple Covarion model, and models that exhibit low-parameter
rate variation across sites, such as the `GTR+$\Gamma$' model. 
An outstanding unsolved problem is whether the widely-used `GTR+$\Gamma$+I' model satisfies the identifiability condition if both the shape parameter and the proportion of invariable sites are unknown \citep{all3}.  Shortly we will describe some models for which  the identifiability condition has been shown to fail.

The kissing condition (\ref{kiseq}) is also relevant to phylogenetics, indeed it applies to {\em all} models  of character evolution used for inferring tree topology. Any  two different trees can produce identical data if the lengths of the interior branches on which the two trees differ shrink to zero and/or the lengths of all (or `most') of the pendant edges grows to infinity  (`site saturation'); these phenomena reflects the geometry of tree-space discussed in \cite{kim} and  \cite{mou} where quite different trees can come arbitrarily close together (`kiss') in terms of the distribution on site patterns they can produce.    This means that the sequence length required to reconstruct a tree correctly by any method,  tends to infinity as the interior branches shrink in length, or as the pendant ones grow.  

Note that this `tree-space' is related to, but quite different from, the tree space described by \cite{bil} -- for example, the latter tree space regards two trees of different topologies as
becoming infinitely far apart as we grow the length of all their branches; however in the tree space here, we regard them as becoming closer together since they are tending to produce exactly the same data (random sequences).   This is illustrated in the following example. 

\begin{example}
{\rm
We can visualise conditions (\ref{ideeq}) and (\ref{kiseq}) by means of a simple but instructive example.   Consider the three rooted binary trees on leaf taxa $1,2,3$, which have branch lengths that satisfy a molecular clock. Let $L$ denotes the length of the interior edge length, and $l$ the length of the shorter pendant edge length, so $L+l$ is the length of the longer pendant edge length.  For the tree $a_1 = 1|23$,
the set $\Theta(a_1)$ is the infinite open first quadrant of the plane: $\{(l,L): L,l>0\}$.

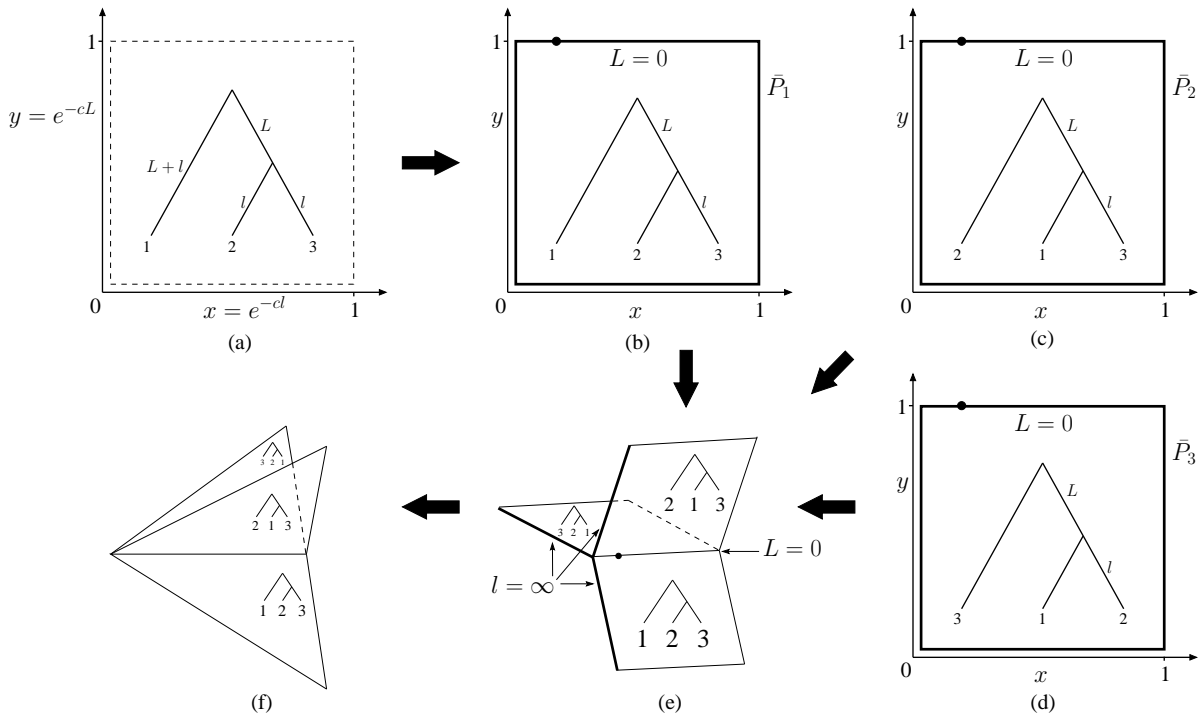
\begin{figure}[ht]
\resizebox{16cm}{!}{
\input{figure.pstex_t}
}
\caption{The `paper-dart' representation of tree space for a Markov process on three taxa subject to a molecular clock (see text for details).}
\label{fig1}
\end{figure}
}
\end{example}

Now consider the function that assigns to $(l, L)$ the probability distribution on site patterns under some model.
 For simple models, such as the $N_r$ model, this function can be described as the composition of two continuous one-to-one and onto  functions.  The first map
associates $(l,L)$ with the vector $(x,y) = (e^{-cl}, e^{-cL})$ where $c$ is a fixed constant (dependent on the model). The image $P_1$  of this map is the open square $(0,1) \times (0,1)$ (Fig. 1(a)).
The second map $\xi: P_1 \rightarrow [0,1]^{|U|}$  sends $(x,y)$ to a probability distribution on site patterns that is determined by the branch lengths associated with the pair $(x,y)$.

For the $N_2$ model, with a uniform probability on the two states at the root,  we have eight site patterns, $(x,y) = (e^{-2l}, e^{-2L})$ (i.e. $c=2$),  and the eight components of $\xi(x,y)$ take just three different values according to whether (i) all three leaves are in the same state, (ii) leaf 1 is in the same state as just one of the other two leaves, or (iii) leaf 1 is in a different state to the other two leaves.  Using standard Hadamard representation (see e.g.  \cite{semste}) these three probabilities, which apply for two, four and two site patterns, respectively, are:
$ \frac{1}{8}(1+x^2+2x^2y^2),  \frac{1}{8}(1-x^2),   \frac{1}{8}(1+x^2-2x^2y^2).$

Moreover, the map $\xi$ extends to $\overline{P_1}$ (the closure of $P_1$ which is the closed square as shown in Fig. 1(b)) and $\overline{p(\Theta(a_1))} = \xi(\overline{P_1})$. Similarly,
for each of the other two trees $a_2= 2|13$ and $a_3=3|12$, we have: $\overline{p(\Theta(a_i))} = \xi(\overline{P_i})$, $i=2,3$ where $\overline{P_2}$ and $\overline{P_3}$ are the corresponding
closed squares for the other two trees (Fig. 1(c,d)).  Each point on the top boundary of these $\overline{P_1}$ (corresponding to $L=0$) has  corresponding points on the top boundary of $\overline{P_2}$
and  $\overline{P_3}$ that induce exactly the same probability distribution on site patterns, and so these three regions `kiss' at each such point (one such shared point is indicated in each region in Fig. 1 (b,c,d)). Thus we can identify
(or glue) these three squares along their top boundary (Fig. 1(e)).  Finally, any point on the front boundary (corresponding to $l = \infty$) leads to the same probability distribution on site patterns - for the $N_r$ model, this would simply assign each of the possible site patterns equal probability.  Thus the whole of this $Y$--shaped part of the complex in Fig. 1(e) is identified to a single point, resulting in the final `paper-dart' representation of the tree space shown in Fig. 1(f) (an example of a `CW complex' in topology).

The main point about this complex is that a one-to-one correspondence can be seen between the points on the `paper dart' and the probability distribution on site patterns that can be induced by  3-taxon trees under a molecular clock where the edge lengths can vary from $0$ to (actual) infinity.  Note that the `spine' of the dart corresponds to the unresolved star tree, while the `head' of the dart corresponds to pendant branches of infinite length.  The `identifiability' condition (\ref{ideeq}) is satisfied since the interior of any one of the three wings does not intersect any other wing (even at the boundary of that other wing), while the
kissing condition (\ref{kiseq}) holds since the wings all touch each other along the central spine (and also, for a different reason, at the front tip).

\subsection{Failure of identifiability for certain CM models}

Note that the identifiability condition generally applies to simple models of site substitution in phylogenetics when $A$ is the set of fully-resolved (binary) phylogenetic trees.  However, it can collapse in three important cases: The first is if  we enlarge  $A$ to the set of all phylogenetic trees (binary and non-binary) on a given set of leaf taxa, since if $a$ has a polytomy, and $b$ is a tree obtained by resolving that polytomy, then:
$$p(\Theta(a)) \cap \Tovb  = p(\Theta(a)) \neq \emptyset.$$
Indeed, even under the CM model, the reconstruction of general (including non-binary) trees using maximum likelihood will not be statistically consistent, since if the generating tree is non-binary, the ML tree will typically be a resolution of this tree for any sequence length (though the lengths of the branches that resolve the polytomy will converge to zero in probability as the sequence length $k$ grows \citep{cha}).
It is possible to consistently reconstruct general (including non-binary) trees, either by modifying ML so as to collapse edges whose length is less than, say $\log(k)/\sqrt{k}$ or by adopting alternative approaches to  tree reconstruction \citep{cho}.

The second situation where the identifiability condition (\ref{ideeq}) may collapse (even when we restrict $A$ to be a the set of fully-resolved trees) is when we have a phylogenetic mixture for certain models.  In this case not only can (\ref{ideeq}) fail, but the examples constructed in (\cite{mat1}) for two-tree mixtures under the $N_2$ model satisfy the stronger violation condition:
\begin{equation}
\label{interinter}
p(\Theta(a)) \cap p(\Theta(b))  \neq \emptyset.
\end{equation}
In the setting of (\cite{mat1}),  $\Theta(a)$ refers to all triples $(\lambda, \lambda', p)$ where $\lambda$ and $\lambda'$ are assignments of positive but finite branch lengths to tree $a$, while $p$ (respectively $1-p$) is the probability that the site evolves under the first (respectively second) set of branch lengths.   Thus (\ref{interinter}) describes the situation in which two
fully resolved trees of differing topology can induce exactly the same probability distribution on site patterns under their particular mixture processes.  To better visualize this, note that in our `paper dart' example earlier, condition (\ref{interinter}) does not occur, but if it did, its geometric interpretation would be that an interior portion (or point)  of one `wing' of the paper dart gets glued to a portion (or point) of a different
`wing'.

A third situation where the identifiability condition (\ref{ideeq}) may collapse is when there is a distribution of rates across sites with two many unknown parameters. Indeed, it was shown in \cite{ras}, that an even stronger violation than (\ref{interinter}) is possible, namely all fully resolved trees on a given set of species can induce the same probability distribution on site patterns for appropriately chosen (but positive and finite) branch lengths in an $N_2$ model and distributions of rates across sites-- in other words:
$$\bigcap_{a \in A} p(\Theta(a)) \neq \emptyset,$$
where $\Theta(a)$ is the set of branch lengths and the parameters describing the distribution of rates across sites.

In cases where the identifiability condition (\ref{ideeq}) fails (i.e. when $p(\Theta(a))\cap \Tova \neq \emptyset $ for  some pair $a \neq b$), for example when a model is `over-parameterized',  we have a useful distinction based on whether or not the overlap of $p(\Theta(a))$ and $\Tova$ has the same or smaller dimension than $p(\Theta(a))$.  In the latter case, although the model fails to
satisfy the strict identifiability condition, it fails only on a subset of $\Theta(a)$ of zero relative volume -- in this case, the tree topology is said to be  {\em generically identifiable} under the model.  The distinction between generic and strict identifiability is important for trying to decide whether SIN is  `theoretically possible but unlikely to occur in practice' or whether there is a reasonable chance of being in a region of parameter space where  we might be unable to  distinguish between two competing trees, even from infinite data.
The distinction has often been overlooked in earlier studies, but is carefully discussed now, particularly as generic identifiability is a notion that sits much more comfortably with current mathematical methods for studying the properties of Markov models based on algebraic geometry  and phylogenetic invariants \citep{all2, all3}.

\section{SIN in the  No Common Mechanism (NCM) setting}
In the NCM model, the $\theta_i$ values may vary in some unknown way.  In particular, we do not assume they are selected i.i.d. from some distribution.
 By analogy with the CM setting, it is tempting to extend the  definition of statistical consistency of a method $\M$  to the NCM setting by the following slight modification:
``For each $a \in A$, and sequence $\theta_i \in \Theta(a)$, the probability that $M$  correctly estimates  $a$ from $(u_1, \ldots, u_k)$, when each $u_i$ is generated independently by the model with parameters $(a, \theta_i)$,  converges to $1$ as $k$ grows.''

However, this condition as stated is too strong: in the tree setting, if the branch lengths grew to infinity (or shrank to zero) sufficiently fast with each observation then accurate tree reconstruction  by any method
for any model can be ruled out (by the Kissing condition).  Nevertheless there are meaningful notions of statistical consistency in the NCM setting, which  generalize the CM definition. 
Recalling that $\Theta(a)$ is  an open subset of Euclidean space, and that a {\em compact} subset of Euclidean space is any subset that is closed and bounded we will consider the following:

{\bf Definition:}  We  will say that a method $\M$ for reconstructing a discrete parameter in a finite set $A$ is {\em statistically consistent for data generated by a NCM model}  if it satisfies the following condition:

   \begin{myindentpar}{1cm}
For each $a \in A$, and every compact subset $C$ of $\Theta(a)$, the probability that $\M$  correctly estimates $a$ from $(u_1, \ldots, u_k)$, when each $u_i$ is generated independently by the model with parameters  
$(a, \theta_i)$, where $\theta_i \in C$,
converges to $1$ as $k$ grows.
   \end{myindentpar}
   This definition is equivalent to the definition of statistical consistency under the CM version of the model if we further insist that all the $\theta_i$ values are equal, and in this case the choice of the compact set $C$
   can be restricted to single points in  $\Theta(a)$.
   
Note also that when we perform ML estimation under CM or NCM we do not require that the $\theta_i$ values associated with $a$  lie within any given compact subset $C$ of $\Theta(a)$; they can take any value in $\Theta(a)$.

\subsection{Which phylogenetic models and methods can lead to SIN?}
The following main result shows that the issue of statistical consistency under NCM is a delicate one, depending on the details of the model and the method. The full mathematical proof of the following results is provided in the Appendix.

\newpage
\begin{theorem}
\label{mainthm}
\mbox{}
\begin{itemize}
\item [{\rm i. }] {\rm [{\bf Inconsistency of ML for NCM--$N_r$ model}]}  Maximum likelihood estimation of fully-resolved tree topology under the NCM--$N_r$ model applied to NCM--$N_r$ data (or even CM--$N_r$ data) is
 statistically inconsistent for any finite $r>1$.  Moreover, no tree reconstruction
method is statistically consistent for NCM--$N_2$ data.
\item [{\rm ii. }] {\rm [{\bf Consistency for the NCM--$N_4$ model}]}  In contrast to part (i), there is a statistically consistent method for inferring fully-resolved tree topology  from data generated by an NCM Jukes-Cantor model.
\item[{\rm iii.}] {\rm[{\bf Consistency for NCM models with a molecular clock}]} Neighbour-joining on uncorrected sequence dissimilarity is a statistically consistent method for inferring fully-resolved tree topology from data generated by an NCM model where each site evolves under its own General Time Reversible (GTR) process (with its own strictly positive rate matrix and branch lengths) provided that, at each site, the branch lengths  are clock-like on the generating tree.
\item [{\rm iv. }] {\rm [{\bf Consistency of ML for NCM--$N_\infty$ model}]} Maximum likelihood estimation of fully-resolved tree topology under the NCM--$N_\infty$  (or even under the CM--$N_\infty$ model)  of NCM--$N_\infty$ data is statistically consistent.
\end{itemize}
\end{theorem}

\section{Measuring SIN, and taking precautions against it}
\label{measure}

We have provided a topological view of tree reconstruction, but there is an equivalent metric view.  To explain this, take  any continuous distance function $d$ on probability distributions on  $U$ (the set of site patterns).  For example, we might take the {\em variational distance}  defined by
$d(p,q) = \frac{1}{2}\sum_{u \in U} |p(u)-q(u)|.$
An alternative way of expressing the identifiability and kissing conditions ((\ref{ideeq}) and (\ref{kiseq})) is to require, for all $a, b \in A$, with $a \neq b$:
\begin{equation}
\label{inf1}
\inf_{\theta' \in \Theta(b)} d(p_{(a,\theta)}, p_{(b, \theta')}) > 0, \mbox{ and }
\end{equation}
\begin{equation}
\label{inf2}
\inf_{\theta \in \Theta(a), \theta' \in \Theta(b)} d(p_{(a,\theta)}, p_{(b, \theta')}) = 0,
\end{equation}
respectively, where `inf' refers to infimum (the minimal value achieved or in the limit).   These are identical conditions to (\ref{ideeq}) and (\ref{kiseq}), respectively, by standard arguments from analysis, based on the Bolzano-Weierstrass Theorem.

One advantage of this metric viewpoint is that a strictly positive value in (\ref{inf1}) not only tells us that  ML is consistent in the CM setting,  but also the magnitude of this value  sets  explicit upper and lower bounds on how much data (sequence length) we require in order to  reconstruct the discrete parameter (tree)  accurately \citep{inv2, inv3}.    The more closely a tree with appropriately chosen branch lengths (or other parameters) can fit the probability distribution on site patterns of a different tree,  the more sequence sites it will take to tell which of the two trees produced the data.

The sequence length required for accurate tree reconstruction (under any CM model) also depends on the number of species being classified. Quantifying this relationship
is particularly challenging  mathematically  (see e.g. \cite{das}).  Various optimal or near-optimal results have been established, which usually require developing a new and clever tree reconstruction method (not because such methods are necessarily better than ML estimation but rather because it has been difficult so far to rigorously establish good bounds on the sequence length required for ML to reconstruct
a large tree accurately).

Much less is known about the sequence length requirements for tree reconstruction under NCM models.  In the case of the $N_\infty$ model (Theorem 1, Part 1(iv)),  the sequence lengths required for accurate tree reconstruction from data generated by an NCM version of the model are essentially the same as for the CM version of the model, provided that in both models, we insist that all  edge lengths are bounded between $(r,s)$ where $0<r\leq s< \frac{1}{2}$. However it seems entirely possible  that for a finite-state $N_r$ model such as the Jukes-Cantor model, the sequence length required for accurate tree reconstruction from NCM--$N_4$ data will be much larger than for a  CM--$N_4$ model with comparable upper and lower bounds on the branch lengths.   If so, this would be another example of where the
two models (finite-state versus infinite-state) have quite different statistical properties. Two other examples are:
\begin{itemize}
\item The sequence length required to resolve a short interior branch of the tree of length $\epsilon$ (from the two alternative tree topologies obtained by swapping branches across the edge)  grows  at the rate $1/\epsilon^2$ for the finite-state model  but just at $1/\epsilon$ for the infinite state model, as $\epsilon \rightarrow 0$ \citep{mossy, inv2}.
\item The sequence length required to determine which of two resolved trees, that classify the same $n$ species, generated the sequence data must grow under the finite-state model, but can be independent of $n$ for the infinite state model \citep{testing}.
\end{itemize}

Returning to the NCM--$N_r$ model of   \cite{tuf},  where branch lengths are allowed to vary freely from site to site, one attempt to avoid the massive over-parameterization of this model is to assume that these different
branch lengths (between characters and across sites) are assigned randomly (i.i.d.) according to some common fixed prior distribution. This `Bayesian NCM' model was explored by \cite{hue}.  As the authors noted, this
model has an interesting property: if one has an underlying $N_r$ model then this `Bayesian NCM model' induces exactly the same probability distribution on site patterns as what might be called the `Ultra-common mechanism model' (UCM), where each character has the same branch length, and these branch lengths are the same {\em across the tree}.  This  formal equivalence between such a tightly constrained model (which would never be used in ordinary phylogenetic practice) and a type of NCM model seems at first a little paradoxical, until it is realized that the
assumption of common fixed prior distribution on branch lengths (across the characters and across the trees) is a very strong `commonality assumption'.  The formal equivalence becomes only approximate for more complex substitution models (technically, this is the result of the rate matrix having multiple nontrivial eigenvalues).

A related but less constrained version of this Bayesian NCM model was developed by  \cite{wu}.  In this model, the tree has underlying branch lengths that are common across the characters, but can vary across the tree, and each site has an intrinsic rate which multiplies the branch lengths across the tree. But in contrast to Olsen's model of allowing this per-site rate to be a free parameter (to be optimized in ML) \cite{wu} assume that this rate parameter is selected i.i.d. from a fixed distribution of rates across sites. For this model, when the underlying substitution process is (say) a Jukes-Cantor model, this intermediate-level Bayesian NCM model assigns exactly the same probability distribution on site patterns as a model in which all the sites evolve i.i.d. under a Jukes-Cantor model (with no rate variation across sites).   As with the UCM model, the formal equivalence becomes only approximate for more complex substitution models, and for the same reasons.  However, in this more general setting, \cite{wu} showed how
a `log-det' transformation gives a statistically consistent way to establish the tree topology from data generated under this model.

The statistical properties of NCM have also been investigated recently from standard model-selection approaches such as AIC \citep{hue2}.   It is clear that NCM can confer higher likelihood scores than a CM model for any data, since one has so much flexibility when choosing the nuisance parameters to get a good fit. Model selection techniques such as AIC (as well as BIC, and other variants) penalize models that are too parameter-rich by subtracting from the log-likelihood of the model a term that depends on the number of parameters \citep{aka}. Under this criterion it seems unlikely that the full-blown NCM model will ever be favored over CM models under AIC.  However it is not entirely clear that the conditions required to justify the AIC criterion extend rigorously to this NCM setting.

\subsection{Can SIN still occur if we  enforce a `no Kissing' condition?}

We conclude this section by pointing out that the Kissing condition (\ref{kiseq}) (or, equivalently (\ref{inf2}))  is not necessarily the cause of SIN in the NCM-setting.   To see this, suppose we constrain the $N_r$ model so
that all the branch lengths in a tree lie between $\epsilon$ and $-\log(\epsilon)$ for some $\epsilon>0$. Let us call this the $N_r^\epsilon$ model.  Then, under the $N_r^\epsilon$ model we have, for different resolved phylogenetic trees $a,b$ on the same set of species:
\begin{equation}
\inf_{\theta \in \Theta(a), \theta' \in \Theta(b)} d(p_{(a,\theta)}, p_{(b, \theta')}) \geq q >0
\end{equation}
where $q = q(\epsilon)$ converges monotonically to $0$ as $\epsilon \rightarrow 0$.
Consequently, for $\epsilon>0$ the Kissing condition fails -- topologically different trees cannot `look' arbitrarily close through the eyes of data produced under a CM model.  However ML estimation under the NCM--$N_r^\epsilon$ model,  can again be statistically inconsistent. Indeed suppose we take any tree and  branch lengths in the interior of a Felsenstein Zone for that tree (i.e. a set of branch lengths where maximum parsimony would converge on an incorrect tree for data produced under the CM--$N_r^\epsilon$ (or NCM--$N_r^\epsilon$ model). Then we can chose $\epsilon>0$ small enough so that the ML estimate of the tree under the NCM--$N_r^\epsilon$ model converges on wrong tree when applied to the CM--$N_r^\epsilon$ data produced from the original tree with its Felsenstein Zone branch lengths.  The formal proof of this claim is given in the Appendix.

\section{Concluding comments}

Making molecular phylogenetic models more `realistic', and  thereby capturing more of the complexities of how DNA evolves   -- across the genome and over different time scales -- usually requires
introducing a number of adjustable parameters. If these parameters can be independently estimated from other data, or if they enter into the model in ways that are not problematic for tree inference (as in Theorem 1(ii-- iv)), or if they follow some common distribution that is described by few if any  unknown parameters),  then  statistically consistent inference of tree topology is achievable.  However, in general, treating branch lengths, and other model parameters as unknown quantities can drive reconstruction methods  to SIN.

Theorem 1 (parts (ii--iv)) provides no real endorsement for NCM models, but it shows that sweeping assumptions that such models must necessarily lead methods to SIN are incorrect.   Such arguments typically proceed as follow:  in NCM models the number of nuisance parameters grows with $k$ and we are unable to estimate them with any precision, thus the usual conditions that suffice for the consistency of maximum likelihood estimation \citep{wal} fail and so the method will be inconsistent. All but the final conclusion of this last sentence are correct -- the failure of a sufficient condition for a statement to be true is not sufficient for it to be false! Indeed, Theorem 1 provides specific cases  where NCM-maximum likelihood estimation under is consistent for certain NCM models. 

Even when statistically consistent methods exist for an NCM model, it is still possible that ML can be statistically inconsistent (this contrasts with what happens in the CM setting, where ML is generally consistent if any consistent method exists).  This leads to a somewhat uncomfortable position for those who wish to provide some statistical justification for the  use of maximum parsimony  as the ML estimator under the NCM model of \cite{tuf} --  By Theorem \ref{mainthm}(i), such a method lives in a state of SIN, yet this could be avoided for this NCM model  if one were to renounce maximum parsimony in favour of a quite different method, such as  one based on linear phylogenetic invariants (the method used in the proof of Theorem \ref{mainthm}(ii)).  

However, this is no strong argument in favour of linear invariants, as they tend to be very inefficient in their ability to extract phylogenetic signal from data \citep{hill94}.  A method based on linear phylogenetic invariants may be guaranteed to converge on the right tree eventually, even under the NCM model, but this may require an astronomical
amount of data.  By contrast, methods such as maximum parsimony appear to be quite efficient at extracting phylogenetic signal when the generating tree branch lengths are some way from those portions of parameter space that lead to inconsistency \citep{hill}.   Thus, although statistical consistency is desirable, it should not over-ride all other considerations -- for example, a powerful method that is consistent in most regions of parameter space would generally be preferred over a statistically consistent method that may requires huge amounts of data to converge.

Of course many of this results in this paper are confined to very simple models (such as the Jukes-Cantor); we have chosen to do this for two reasons:  firstly, they are sufficient to demonstrate that
even with very simple models, all possibilities (consistency and SIN) are possible given slight tweaks of the assumptions or method; secondly, the analysis of more complex models is beyond the scope of this paper, but would be a worthy objective for future work.

In summary, the question of whether one is prone to SIN by adopting a particular NCM model and a particular method of inference has a more complex answer than in the CM setting -- it depends subtly  on the details of the model and on the method.   The full-blown generality of the NCM of \cite{tuf} is unnecessarily over-parameterized for most data, being a model that was developed to prove a formal equivalence between methods rather than as a model of choice.  Far from being a justification of MP, its plethora of ever-growing parameters would surely have not seemed `parsimonious' to William of Occam.   At the other extreme are simple attempts to include NCM within a Bayesian framework; these avoid SIN, but at the price of forcing the NCM model into a CM straightjacket by viewing the parameters as samples from a common underlying prior.   
Between these extremes there would seem to be an endless variety of possibilities.
The development of carefully constrained yet parameter-rich models, guided by model selection criteria, and which recognize that characters evolve under different processes dependent on their biochemistry,  will surely play a significant  role in future phylogenetic methodology.

\section{Acknowledgments}
I wish to thank the Royal Society of New Zealand for funding under its James Cook Fellowship scheme, 
Dietrich Radel for typesetting Figure 1, and Elliott Sober and Peter Lockhart for some helpful comments on an earlier version of this manuscript. 

\bibliographystyle{sysbio}
\bibliography{SIN}
\newpage
\section{Appendix: Technical details}

The following Lemmas are required in the proof of Theorem 1. 

\begin{lemma}[Azuma's inequality]
There are several variants of this inequality (see for example \cite{gri}) here we give a special case of a more general version.
Suppose that $X_1, \ldots, X_k$ are independent random variables taking values in some set $S$, and that $Y = f(X_1, \ldots, X_k)$ where $f: S^k \rightarrow \RR$ is any function with the property that
$|f(y_1, \dots, y_k)-f(y_1',\ldots, y_k')| \leq c$
whenever $y_i' = y_i$ for all but one value of $i$. Then $\PP(Y - \EE[Y] \geq x)$ and $\PP(Y-\EE[Y] \leq -x)$ are each less or equal to $\exp(-x^2/2c^2k).$
\end{lemma}

   \begin{lemma}
   \label{lem1}
   A method $\M$  for estimating the discrete parameter $a \in A$ from sequences of observations in $U$  is statistically consistent under a NCM-M model if it satisfies the following property:
 for each $a \in A$, there is a nest family  $\Tint_k(a), k = 1,2 \ldots$, of increasing  open subsets of  Euclidean space
 with  $\Theta(a)= \bigcup_{k=1}^\infty \Tint_i(a)$, so that the following condition holds: the  probability that $\M$  correctly estimates element $a$ from $(u_1, \ldots, u_k)$, whenever each $u_i$ is generated by
 $(a, \theta_i)$,  with $\theta_i \in \Tint_k(a)$, converges to $1$ as $k$ grows.
  \end{lemma}
  \begin{proof}
Suppose $C$ is a compact subset of $\Theta(a)$.
Then $C \cap \Tint_k(a), i \geq 1$ is an open cover of $C$. Since $C$ is compact, $C$ is equal to the union of finitely many of the sets $C \cap \Tint_k(a)$, and since the
sets $\Tint_k(a)$, $k \geq 1$ are nested,  a value of $k=k_1$ exists for which $C \subset \Tint_{k_1}(a)$.  By the hypothesis of the Lemma, the event that $\M$ correctly
returns any $a$ from $(u_1,\dots, u_k)$ when each $u_i$ is generated by $(a, \theta_i)$, where $\theta_i \in \Tint_k(a)$, has probability that converges to $1$ as $k$ grows.    Since $C \subseteq \Tint_{k_{1}} \subseteq \Tint_{k}$ for
$k \geq k_1$,
restricting $\theta_i$ to lie in $C$ ensures that the probability $\M$ correctly returns $a$ from $(u_1,\dots, u_k)$ when each $u_i$ is generated by $(a, \theta_i)$, where $\theta_i \in C$ also converges to $1$ as $k$ grows.

\end{proof}

\subsection{Proof of Theorem 1}

{\em Part (i):}  ML estimation under the NCM--$N_r$ model applied to {\em any} sequence of $r$--state characters returns the same tree(s) as maximum parsimony \citep{tuf}. This latter method was shown to be statistically inconsistent for
CM--$N_2$ data \citep{fels}  and, more generally,  for CM--$N_r$  data for $r \geq 2$ by later authors (see \cite{shu}, and the references therein) even for trees on four species.  Since the CM--$N_r$ model is just a submodel of the NCM--$N_r$ model, both assertions in the first claim of Part (i) follows. Specifically, we can take $a$ to be any resolved binary tree on four leaves, and $\Theta(a)$ to be $(0,\infty)^5$ and
select the $\theta_i$ values all to be equal to a choice of branch lengths $\theta \in (0,\infty)^5$ for which maximum parsimony (and thereby ML under NCM--$N_r$) converges on an incorrect tree.

The proof of the second claim, that concerning the NCM--$N_2$ model, follows directly from the examples in
 \cite{mat2}.

For parts (ii--iv) we  will establish the statistical consistency of  various methods by establishing the property described in the following lemma.

{\em Part (ii):}  The proof relies  on the existence of certain linear phylogenetic invariants for the Jukes-Cantor model
  (the existence of such invariants for models that include the Jukes-Cantor was desribed by \cite{lake}).  In particular, from Theorem 1 (part 5) of \cite{fu},  any binary phylogenetic tree $T$ has an associated   function $L_T$ of the site pattern frequencies, such that (i) $L_T({\bf p}) =0$ where  ${\bf p}$ is the probability vector of site patterns generated by $T$ under any assignment of branch lengths,
 and (ii) for any binary phylogenetic tree $T'$ that is different from  $T$, but has the same leaf set, we have $L_{T'}({\bf p}) \geq f_T(u,v)>0$
where $u$ is the shortest branch length,  $v$ is the largest branch length, and $f$ is a continuous function that has the following two properties:
\begin{itemize}
\item for all $u>0$, $f$ is monotone decreasing in $v$ and converges to $0$ as $v$ tends to infinity;
\item for all $v>0$, $f$ is monotone increasing in $u$ and converges to $0$ as $u$ tends to $0$.
\end{itemize}
Although these are all the properties of $f$ we require for the rest of the proof we provide  an explicit description of $f$ is provided as follows:

For a binary tree $T$ on four leaves (a quartet tree) with topology $ij|kl$, a linear invariant $L_T$ for the Jukes-Cantor model was described in \cite{fu} with the properties described
in the proof of Part (ii) of Theorem~\ref{mainthm}, with:
$f(u,v) = \exp(-\frac{4}{3}S)\cdot\left(1-\exp(-\frac{8}{3}l)\right)$ where
$S$ is the sum of the lengths of the four pendant edges, and $l$ is the length of the interior edge.
For a binary tree $T$  with more than four leaves, select any collection $Q$ of quartet trees that define $T$ (i.e. for which $T$ is the unique tree that displays those quartet trees) and take the sum of the
linear invariants just described to give a linear invariant $L_T$. Notice that $L_T$ also satisfies the condition described in the proof of Part (ii) of Theorem~\ref{mainthm}
by taking $f(u,v)$ to be the function:
$\exp(-\frac{4}{3}(2n-4)v)\cdot\left(1-\exp(-\frac{8}{3}u)\right)$
where  $n$ is the number of species (so the number of edges in the pendant edges of any induced quartet tree is at most $(2n-4)$). This is the function $f$ promised.

Returning to the proof of Part (ii), for a site $s$,  let ${\bf X}^s$ be the $4^n$--dimensional vector, indexed by site patterns, that takes the value $1$ for the site pattern observed at site $s$ and $0$ otherwise, and let
$\hat{\bf x}^{(k)} = \frac{1}{k} \sum_{s=1}^k {\bf X}^s.$   Consider the following tree reconstruction method  $\M:  \mbox{Select the binary phylogenetic tree $T$ that minimizes $L_T(\hat{\bf x}^{(k)})$}. $

We will show that
$\M$ is statistically consistent under a NCM--$N_4$ model, by ensuring that the branch lengths in the generating tree at site $i$ lie between $\epsilon_i$ and $L_i$, where these two sequences
converge monotonically to zero and to infinity (respectively), sufficiently slowly with $i$.

To this end, suppose $k$ sites evolve on a fully-resolved tree $T$ under a NCM--$N_4$ model.  Let ${\bf p}^s = \EE[{\bf X}^s]$, the vector of probabilities of the different site patterns at site $s$,
 and let   ${\bf p}^{(k)}: = \frac{1}{k} \sum_{s=1}^k {\bf p}^s$.  By the invariant property of $L_T$,  we have $L_T({\bf p}^s) = 0$ for all $s$, and since $L_T$ is linear, it follows that:
\begin{equation}
\label{L1eq}
\EE[L_T(\hat{\bf x}^{(k)})] = L_T( \EE[\hat{\bf x}^{(k)}] ) = L_T( {\bf p}^{(k)}) = 0 \mbox{ for all } k\geq 1.
\end{equation}

Similarly, for any fully-resolved  phylogenetic tree $T'$ that is different from $T$, but has the same leaf set, we have:
\begin{equation}
\label{L2eq}
\EE[L_{T'}(\hat{\bf x}^{(k)})] = L_{T'}({\bf p}^{(k)}) > f(\epsilon_k, L_k).
\end{equation}
From the continuity of $f$ and its other listed properties, we can allow $\epsilon_k$ to tend to zero and $L_k$ to tend to infinity sufficiently slowly (with increasing $k$) that the following condition is satisfied:
\begin{equation}
\label{epsy}
\lim_{k \rightarrow \infty} k\cdot f^2(\epsilon_k, L_k) \rightarrow \infty.
\end{equation}

Now, since the  ${\bf X}^s: s =1, \ldots, k$ are independent random variables, the Azuma inequality combined with
(\ref{L1eq}) and (\ref{epsy})  gives:
$\lim_{k \rightarrow \infty} \PP\left(L_{T}(\hat{\bf x}^{(k)}) > \frac{1}{2}f(\epsilon_k, L_k)\right) =0;$
while for any alternative fully-resolved phylogenetic tree $T'$ ($\neq T$), Eqns. (\ref{L2eq}) and (\ref{epsy}) give:
$\lim_{k \rightarrow \infty} \PP\left(L_{T'}(\hat{\bf x}^{(k)}) < \frac{1}{2}f(\epsilon_k, L_k)\right) =0.$
Combining these two last equations gives:
$\lim_{k \rightarrow \infty}\PP\left(L_{T}(\hat{\bf x}^{(k)}) < L_{T'}(\hat{\bf x}^{(k)}) \right) =1,$
and so $$ \lim_{k \rightarrow \infty}\PP\left(L_{T}(\hat{\bf x}^{(k)}) < L_{T'}(\hat{\bf x}^{(k)}) \mbox{ for all } T'\neq T \right) =1.$$

By Lemma~\ref{lem1}, this implies that method $\M$ is statistically consistent under the model described.

{\em Part (iii):}  Let $d^{(k)}_{ij}$ denote the proportion of the $k$ sites on which species $i, j$ disagree and let $\mu^{(k)}_{ij} = \EE[d^{(k)}_{ij}]$.  Thus $d^{(k)}_{ij} = \frac{1}{k} \sum_{s=1}^k \xi^{ij}_s$ where $\xi^{ij}_s$ takes the value $1$ if sequences $i$ and $j$ differ at site $s$, and $0$ otherwise.
By the standard theory of reversible $r$--state Markov processes, combined with the molecular clock hypothesis,  for any two species $x,y$,  we can write
\begin{equation}
\label{eqy}
\EE[\xi^{xy}_s] =(1 - \sum_{i=1}^r \pi_{s,i}^2) + \sum_{j=1}^{r-1} \alpha_{s,j}e^{-2\beta_{j,s} t_{xy}}
\end{equation}
 where:
\begin{itemize}
\item
$\pi_{s,i}$ is the  vector of equilibrium base frequency of base $i$ at site $s$;
\item
$-\beta_{s,j}$ are the non-zero eigenvalues of the GTR rate matrix at site $s$;
\item
the coefficients $\alpha_{s,j}$ are positive (and determined by the
eigenvalues of the GTR matrix at site $s$, along with the $\pi_{s,i}$ values);
\item
$t_{xy}$ is the time from when  species $x$ and $y$ diverged in the tree to the present.
\end{itemize}

Consequently, if the generating tree $T$ resolves the  triplet of species $i,j,l$ as the rooted tree $ij|l$ then:
\begin{equation}
\label{xieq}
\EE[\xi^{il}_s] = \EE[\xi^{jl}_s] > \EE[\xi^{ij}_s], \mbox{ and } \EE[\xi^{il}_s] -  \EE[\xi^{ij}_s] \geq g(u_s, v_s)
\end{equation}
 where $g$ is a continuous function that has the same properties as $f$ in the previous proof  and where:
 \begin{itemize}
\item  $u_s$ is the sum of the branch lengths on the path between the least common ancestor of $i,l$ and the least common ancestor of $i,j$ at site $s$ times the substitution rate at site $s$;
\item $v_s$ is the sum of the branch lengths between the root and any leaf  multiplied by the largest magnitude of any eigenvalue of the GTR matrix at site $s$.
\end{itemize}

An explicit description of the function $g$ is as follows: From (\ref{eqy}) we have
$\EE[\xi^{il}_s] -  \EE[\xi^{ij}_s] = \sum_{j=1}^{r-1} \alpha_{s,j}\left(\exp(-2\beta_{j,s} t_{il}) -\exp(-2\beta_{j,s} t_{ij})\right),$
and, using the identity $e^{-x}- e^{-y} = e^{-y}(e^{y-x}-1) \geq e^{-y}(y-x)$ for $0<x<y$, we have:
$$\EE[\xi^{il}_s] -  \EE[\xi^{ij}_s] \geq 2\sum_{j=1}^{r-1} \alpha_{s,j}\beta_{j,s} \exp(-2\beta_{j,s} t_{il})\cdot  (t_{il}-t_{ij}).$$
Now the term $\sum_{j=1}^{r-1} \alpha_{s,j}\beta_{j,s}$ is the substitution rate at site $s$, and so we can set $g(u,v) = 2u_se^{-2v_s}.$

Thus, if we let $\mu^{(k)}_{ij} = \EE[d^{(k)}_{ij}]$ then
\begin{equation}
\label{mueq}
\mu^{(k)}_{il} = \mu^{(k)}_{jl} \mbox{ and } \mu^{(k)}_{il}-\mu^{(k)}_{ij}> g(\epsilon_k, L_k),
\end{equation}
 where $\epsilon_k = \min\{u_s:  1\leq s \leq k\}$ and $L_k= \max\{v_s: 1\leq s \leq k\}$.

As in the previous proof, by the continuity of $g$ and its other listed properties, we can allow $\epsilon_k$ to tend to zero and $L_k$ to tend to infinity sufficiently slowly (with increasing $k$) that  $\lim_{k \rightarrow \infty} k\cdot g^2(\epsilon_k, L_k) \rightarrow \infty.$
Then by Azuma's inequality:
\begin{equation}
\label{asu}
\lim_{k \rightarrow \infty}  \PP\left(\max_{ij} |d^{(k)}_{ij} - \mu^{(k)}_{ij}| \geq \frac{1}{2}g(\epsilon_k, L_k)\right)=0.
\end{equation}
Note that,  by (\ref{mueq}) the $\mu$ values are additive on $T$ and each interior edge has a branch length of at least $g(\epsilon_k, L_k)$.
  We can thus invoke the `safety radius' result of \cite{att} which guarantees that neighbor-joining applied to the matrix of  $d^{(k)}_{ij}$ values, for all pairs $i,j$, will return $T$ provided that each pairwise distance $d^{(k)}_{ij}$ differs from $\mu^{(k)}_{ij}$ by at most  $\frac{1}{2}g(\epsilon_k, L_k)$.  This last event has probability converging to $1$ as $k$ grows by (\ref{asu}) and so Part (iii) now follows from Lemma~\ref{lem1}.

{\em Part (iv):}
  For any data consisting of a sequences of characters on a set of species, the only phylogenetic trees that have positive likelihood under the NCM--$N_\infty$ model are those on which the data are homoplasy-free (i.e. require no reverse or convergent evolutionary events).   Thus,  it suffices to show that for $k$ characters generated by a NCM--$N_\infty$ model on $T$, the probability that $T$ is the only phylogenetic tree for the given species
on which these characters are homoplasy-free converges to $1$ as $k \rightarrow \infty$. Following  \cite{war}, it suffices to show  that the following event $E_k$ has probability converging to $1$ as $k$ grows: $E_k$ is the
event that for each induced quartet tree $ab|cd = T|\{a,b,c,d\}$ of $T$, at least one of the $k$ characters assigns the same state to $a$ and $b$, and the same state to $c$ and $d$,  and with  these two states being different.
By the independence assumption between changes on different edges in the $N_\infty$ model, and by the Bonferroni inequality, we have:
\begin{equation}
\label{binoeq}
\PP(E_k) \geq 1 - \binom{n}{4} \cdot \prod_{s=1}^k (1- p_sq_s^4)
\end{equation}
where $p_s$ (respectively $q_s$) is the smallest substitution probability on an edge (respectively the largest substitution probability on a path) for the process that generates site $s$.
Thus, provided that the branch lengths at site $s$ are bounded between $(\epsilon_k, L_k)$ where $\epsilon_k$ converges to $0$ sufficiently slowly, and that $L_k$ converges to infinity sufficiently slowly (with increasing $k$), then
$\lim_{k \rightarrow \infty} \PP(E_k)=1$, by (\ref{binoeq}).
Part (iv) now follows from Lemma~\ref{lem1}.

\subsection{Proof that ML under an $\epsilon$--constrained NCM can be statistically inconsistent}

For ${\bf u}=(u_1,\ldots, u_k) \in U^k$ and a fully resolved tree $a$, let $L_a({\bf u}) $ be the log of the maximum likelihood value of the data ${\bf u}$
under the NCM--$N_r$ model. From \cite{tuf} we have:
\begin{equation}
\label{par1}
L_a({\bf u}) = -(l({\bf u}, a) + k)\cdot \log(r),
\end{equation}
where $l({\bf u}, a)$ is the  parsimony score of ${\bf u}$ on $a$.
Similarly, for $\epsilon>0$, let $L_a^\epsilon({\bf u}) $ be the log of the maximum likelihood value of the data ${\bf u}$
under the constrained NCM--$N_r$ model on tree $a$ in which each branch length is required to lie between $\epsilon$ and $-\log(\epsilon)$.  Clearly, $L_a^\epsilon({\bf u})  \leq L_a({\bf u})$.

Consider the following way to `prune' branch lengths in any tree $c$ which associates to each vector of branch lengths $\theta$ a corresponding set of branch lengths $\theta^\epsilon$ that
satisfy the $\epsilon$ constraint:  For each branch length shorter than $\epsilon$ re-set that
branch length to $\epsilon$, and for each branch length larger than $-\log(\epsilon)$ reset it to $-\log(\epsilon)$.  This transformation
$\theta \mapsto \theta^\epsilon$ enjoys the following property for the $N_r$ model:
For any site pattern $u \in U$ we have:
$$\PP(u|c, \theta^\epsilon) \geq \PP(u|c, \theta) - O(\epsilon),$$
where $\PP(u|c, \theta^{(\epsilon)})$ is the probability of generating $u$ on tree $c$ with branch lengths $\theta^{(\epsilon)}$ and
where $O(\epsilon)$ is a term that depends just on $\epsilon$ and the number of leaves in the tree, and which tends to zero as $\epsilon \rightarrow 0$.
It follows that:
\begin{equation}
\label{par4}
\frac{1}{k}L_c^\epsilon({\bf u}) \geq \frac{1}{k} L_c({\bf u}) - O(\epsilon).
\end{equation}

Now, by elementary algebra:
\begin{equation}
\label{oxt}
\frac{1}{k}(L_a^\epsilon({\bf u}) - L_b^\epsilon({\bf u}))= \Delta_1+\Delta_2+\Delta_3
\end{equation}
where:
$$\Delta_1= \frac{1}{k}(L_a^\epsilon({\bf u}) - L_a({\bf u})) \leq 0,$$
$$\Delta_2=\frac{1}{k}(L_a({\bf u}) - L_b({\bf u})), \mbox{ and } \Delta_3=\frac{1}{k}(L_b({\bf u}) - L_b^\epsilon({\bf u})).$$
Now, for any two trees $a, b$, Eqn. (\ref{par1}) gives:
\begin{equation}
\label{felso}
\Delta_2 =\left (\frac{l({\bf u}, b)}{k} - \frac{l({\bf u}, a)}{k}\right)\cdot \log(r).
\end{equation}
If ${\bf u}$ is generated by a CM--$N_r$ on $a$ with branch lengths in the interior of the Felsenstein Zone (a region of branch lengths for tree $a$ where maximum parsimony converges on an incorrect tree) then, for tree $b$ having a different topology from $a$ the term
$\frac{l({\bf u}, b)}{k} - \frac{l({\bf u}, a)}{k}$ in (\ref{felso}) converges in probability to a  negative constant $-C$ (the actual value of which is dependent on the branch lengths used in the Felsenstein Zone setting).

Regarding $\Delta_3$, we can apply Inequality~(\ref{par4}) for $c=b$ and select $\epsilon$ sufficient small (but strictly positive) so that (i) the branch lengths used in the Felsenstein Zone setting are all
greater than $\epsilon$ and less that $-\log(\epsilon)$ and (ii) the $O(\epsilon)$ term in (\ref{par4}) is less than $C\log(r)$ and so, for all $k \geq 1$ and all ${\bf u} \in U^k$:
\begin{equation}
\label{par5}
\frac{1}{k}(L_b({\bf u}) -  L_b^\epsilon({\bf u})) \leq  \frac{1}{2}C\log(r).
\end{equation}
Thus, from (\ref{oxt}) and (\ref{par5}) we have:
\begin{equation}
\label{oxy}
\frac{1}{k}(L_a^\epsilon({\bf u}) - L_b^\epsilon({\bf u})) \leq \Delta_2 + \frac{1}{2}C\log(r),
\end{equation}
for $\epsilon>0$ sufficiently small.

Since $\Delta_2$ converges in probability to $-C\log(r)$ with increasing $k$, it follows from (\ref{oxy}) that the probability that $ \frac{1}{k}(L_a^\epsilon({\bf u}) - L_b^\epsilon({\bf u}))$ is negative when $(u_1, \ldots, u_k)$ is generated under the CM--$N_r$ model tends to $1$ as $k$ grows.
That is, ML estimation under an NCM--$N_r^\epsilon$ model is
statistically inconsistent,  for data generated under the CM--$N_r^\epsilon$ model (or a NCM--$N_r^\epsilon$ model) when the branch lengths lie within the Felsenstein Zone for tree $a$, and $\epsilon$ is chosen sufficiently small (relative to those
branch lengths).

\newpage

\end{document}

%% file: figure.pstex_t
\begin{picture}(0,0)%
\epsfig{file=figure.pstex}%
\end{picture}%
\setlength{\unitlength}{4144sp}%
\begingroup\makeatletter\ifx\SetFigFont\undefined%
\gdef\SetFigFont#1#2#3#4#5{%
  \reset@font\fontsize{#1}{#2pt}%
  \fontfamily{#3}\fontseries{#4}\fontshape{#5}%
  \selectfont}%
\fi\endgroup%
\begin{picture}(13206,7807)(181,-7216)
\put(2339,-2851){\makebox(0,0)[lb]{\smash{\SetFigFont{17}{20.4}{\rmdefault}{\mddefault}{\updefault}\special{ps: gsave 0 0 0 setrgbcolor}$x=e^{-cl}$\special{ps: grestore}}}}
\put(2969,-781){\makebox(0,0)[lb]{\smash{\SetFigFont{12}{14.4}{\rmdefault}{\mddefault}{\updefault}\special{ps: gsave 0 0 0 setrgbcolor}$L$\special{ps: grestore}}}}
\put(1709,-1231){\makebox(0,0)[lb]{\smash{\SetFigFont{12}{14.4}{\rmdefault}{\mddefault}{\updefault}\special{ps: gsave 0 0 0 setrgbcolor}$L+l$\special{ps: grestore}}}}
\put(13094,-376){\makebox(0,0)[lb]{\smash{\SetFigFont{17}{20.4}{\rmdefault}{\mddefault}{\updefault}\special{ps: gsave 0 0 0 setrgbcolor}$\bar{P_2}$\special{ps: grestore}}}}
\put(11925,-4838){\makebox(0,0)[lb]{\smash{\SetFigFont{12}{14.4}{\rmdefault}{\mddefault}{\updefault}\special{ps: gsave 0 0 0 setrgbcolor}$L$\special{ps: grestore}}}}
\put(13095,-4433){\makebox(0,0)[lb]{\smash{\SetFigFont{17}{20.4}{\rmdefault}{\mddefault}{\updefault}\special{ps: gsave 0 0 0 setrgbcolor}$\bar{P_3}$\special{ps: grestore}}}}
\put(7424,-781){\makebox(0,0)[lb]{\smash{\SetFigFont{12}{14.4}{\rmdefault}{\mddefault}{\updefault}\special{ps: gsave 0 0 0 setrgbcolor}$L$\special{ps: grestore}}}}
\put(11924,-781){\makebox(0,0)[lb]{\smash{\SetFigFont{12}{14.4}{\rmdefault}{\mddefault}{\updefault}\special{ps: gsave 0 0 0 setrgbcolor}$L$\special{ps: grestore}}}}
\put(3420,-1636){\makebox(0,0)[lb]{\smash{\SetFigFont{12}{14.4}{\rmdefault}{\mddefault}{\updefault}\special{ps: gsave 0 0 0 setrgbcolor}$l$\special{ps: grestore}}}}
\put(7875,-1636){\makebox(0,0)[lb]{\smash{\SetFigFont{12}{14.4}{\rmdefault}{\mddefault}{\updefault}\special{ps: gsave 0 0 0 setrgbcolor}$l$\special{ps: grestore}}}}
\put(12375,-1636){\makebox(0,0)[lb]{\smash{\SetFigFont{12}{14.4}{\rmdefault}{\mddefault}{\updefault}\special{ps: gsave 0 0 0 setrgbcolor}$l$\special{ps: grestore}}}}
\put(12375,-5686){\makebox(0,0)[lb]{\smash{\SetFigFont{12}{14.4}{\rmdefault}{\mddefault}{\updefault}\special{ps: gsave 0 0 0 setrgbcolor}$l$\special{ps: grestore}}}}
\put(2755,-1636){\makebox(0,0)[lb]{\smash{\SetFigFont{12}{14.4}{\rmdefault}{\mddefault}{\updefault}\special{ps: gsave 0 0 0 setrgbcolor}$l$\special{ps: grestore}}}}
\put(8594,-376){\makebox(0,0)[lb]{\smash{\SetFigFont{17}{20.4}{\rmdefault}{\mddefault}{\updefault}\special{ps: gsave 0 0 0 setrgbcolor}$\bar{P_1}$\special{ps: grestore}}}}
\put(181,-691){\makebox(0,0)[lb]{\smash{\SetFigFont{17}{20.4}{\rmdefault}{\mddefault}{\updefault}\special{ps: gsave 0 0 0 setrgbcolor}$y=e^{-cL}$\special{ps: grestore}}}}
\put(10035,-4741){\makebox(0,0)[lb]{\smash{\SetFigFont{17}{20.4}{\rmdefault}{\mddefault}{\updefault}\special{ps: gsave 0 0 0 setrgbcolor}$y$\special{ps: grestore}}}}
\put(11566,-6901){\makebox(0,0)[lb]{\smash{\SetFigFont{17}{20.4}{\rmdefault}{\mddefault}{\updefault}\special{ps: gsave 0 0 0 setrgbcolor}$x$\special{ps: grestore}}}}
\put(11566,-2851){\makebox(0,0)[lb]{\smash{\SetFigFont{17}{20.4}{\rmdefault}{\mddefault}{\updefault}\special{ps: gsave 0 0 0 setrgbcolor}$x$\special{ps: grestore}}}}
\put(7066,-2851){\makebox(0,0)[lb]{\smash{\SetFigFont{17}{20.4}{\rmdefault}{\mddefault}{\updefault}\special{ps: gsave 0 0 0 setrgbcolor}$x$\special{ps: grestore}}}}
\put(10035,-691){\makebox(0,0)[lb]{\smash{\SetFigFont{17}{20.4}{\rmdefault}{\mddefault}{\updefault}\special{ps: gsave 0 0 0 setrgbcolor}$y$\special{ps: grestore}}}}
\put(5536,-691){\makebox(0,0)[lb]{\smash{\SetFigFont{17}{20.4}{\rmdefault}{\mddefault}{\updefault}\special{ps: gsave 0 0 0 setrgbcolor}$y$\special{ps: grestore}}}}
\put(6841,-61){\makebox(0,0)[lb]{\smash{\SetFigFont{17}{20.4}{\rmdefault}{\mddefault}{\updefault}\special{ps: gsave 0 0 0 setrgbcolor}$L=0$\special{ps: grestore}}}}
\put(11341,-61){\makebox(0,0)[lb]{\smash{\SetFigFont{17}{20.4}{\rmdefault}{\mddefault}{\updefault}\special{ps: gsave 0 0 0 setrgbcolor}$L=0$\special{ps: grestore}}}}
\put(11341,-4111){\makebox(0,0)[lb]{\smash{\SetFigFont{17}{20.4}{\rmdefault}{\mddefault}{\updefault}\special{ps: gsave 0 0 0 setrgbcolor}$L=0$\special{ps: grestore}}}}
\put(8551,-5506){\makebox(0,0)[lb]{\smash{\SetFigFont{17}{20.4}{\rmdefault}{\mddefault}{\updefault}\special{ps: gsave 0 0 0 setrgbcolor}$L=0$\special{ps: grestore}}}}
\put(5536,-5911){\makebox(0,0)[lb]{\smash{\SetFigFont{17}{20.4}{\rmdefault}{\mddefault}{\updefault}\special{ps: gsave 0 0 0 setrgbcolor}$l=\infty$\special{ps: grestore}}}}
\end{picture}